\def\lncs{0}

\documentclass[11pt]{article}
\usepackage[colorlinks]{hyperref}
\hypersetup{linkcolor=[rgb]{0,0.4,0.7},filecolor=[rgb]{0,0.4,0.7},citecolor=[rgb]{0,0.4,0.7},urlcolor=[rgb]{0,0,0.7}}
\usepackage{xspace}
\usepackage{fullpage}
\usepackage{boxedminipage}
\usepackage[boxed]{algorithm}
\usepackage{epigraph}
\usepackage[sc]{mathpazo}
\usepackage{amsmath}
\usepackage{accents}
\usepackage{caption}
\usepackage{subcaption}
\usepackage[normalem]{ulem}
\usepackage{amsthm}
\usepackage{amssymb}
\usepackage{amsfonts}

\ifnum\lncs=0
\usepackage{amsthm}
\fi
\usepackage{amsmath,amssymb}

\ifnum\lncs=0
\newtheorem{theorem}{Theorem}[section]

\newtheorem{definition}[theorem]{Definition}

\newtheorem{lemma}[theorem]{Lemma}

\newtheorem{claim}[theorem]{Claim}

\newtheorem{corollary}[theorem]{Corollary}

\fi

\usepackage{tikz}
\usepackage{relsize}
\usepackage{ctable}
\usepackage{cleveref,aliascnt}

\newcommand{\eylon}[1]{\begin{center} \framebox{ \parbox{ 13cm }
{\textcolor[rgb]{.8,0,0}{{\bf Eylon:} #1}}} \end{center}}
\newcommand{\omri}[1]{\begin{center} \framebox{ \parbox{ 13cm }
{\textcolor[rgb]{.8,0,0}{{\bf Omri:} #1}}} \end{center}}

\newcommand{\E}[1]{{\mathbb E}\! \left[ {#1} \right]}

\newcommand{\Bin}{\textsf{Bin}}

\newcommand{\N}{{\mathbb{N}}}

\newcommand{\Var}{\mathsf{Var}}

\newcommand{\Sampler}{\mathsf{Sampler}}
\newcommand{\Adversary}{\mathsf{Adversary}}
\newcommand{\UniformSample}{\mathsf{BernoulliSample}}
\newcommand{\ReservoirSample}{\mathsf{ReservoirSample}}
\newcommand{\AdaptiveGame}{\mathsf{AdaptiveGame}}
\newcommand{\ContinuousAdaptiveGame}{\mathsf{ContinuousAdaptiveGame}}

\newcommand{\ith}[1]{{#1}\textsuperscript{th}}

\newcommand{\ignore}[1]{}

\newenvironment{boxfig}[2]{\begin{figure}[#1]\fbox{\begin{minipage}{\linewidth}
				\vspace{0.2em}
				\makebox[0.025\linewidth]{}
				\begin{minipage}{0.95\linewidth}
					{{
							#2 }}
				\end{minipage}
				\vspace{0.2em}
	\end{minipage}}}{\end{figure}}

\newcommand{\pprotocol}[4]{
	\begin{boxfig}{h!}{
			\begin{center}
				\textbf{#1}
			\end{center}
			#4
			\vspace{0.2em} } \caption{\label{#3} #2}
	\end{boxfig}
}

\newcommand{\protocol}[4]{
	\pprotocol{#1}{#2}{#3}{#4} }

\title{The Adversarial Robustness of Sampling}

\author{Omri Ben-Eliezer\thanks{Blavatnik School of
		Computer Science, Tel Aviv University, Tel Aviv 69978, Israel.}
		\and Eylon Yogev\thanks{Department of Computer Science, Technion, 
		Haifa, Israel. Supported by the European Union's Horizon 2020 research 
		and innovation program under grant agreement no.\ 742754, and by a 
		grant from the Israel Science Foundation (no.\ 950/16).}}

\date{}

\begin{document}
\maketitle
\thispagestyle{empty}

\begin{abstract}

Random sampling is a fundamental primitive in modern algorithms, statistics, and machine learning, used as a generic method to obtain a small yet ``representative'' subset of the data. In this work, we investigate the robustness of sampling against adaptive adversarial attacks in a streaming setting: An adversary sends a stream of elements from a universe $U$ to a sampling algorithm (e.g., Bernoulli sampling or 
reservoir sampling), with the goal of making the sample ``very unrepresentative'' of the underlying data stream. The adversary is fully adaptive in the sense that it knows the exact content of the sample at any given point along the stream, and can choose which element to send next accordingly, in an online manner.

Well-known results in the static setting indicate that if the full stream is chosen in advance (non-adaptively), then a random sample of size $\Omega(d / \varepsilon^2)$ is an $\varepsilon$-approximation of the full data with good probability, where $d$ is the VC-dimension of the underlying set system $(U, \mathcal{R})$. 
Does this sample size suffice for robustness against an adaptive adversary? 
The simplistic answer is \emph{negative}: We demonstrate a set system where a constant sample size (corresponding to a VC-dimension of $1$) suffices in the static setting, yet an adaptive adversary can make the sample very unrepresentative, as long as the sample size is (strongly) sublinear in the stream length, using a simple and easy-to-implement attack. 

However, this attack is ``theoretical only'', requiring the set system size to (essentially) be exponential in the stream length. This is not a coincidence: We show that in order to make the sampling algorithm robust against adaptive adversaries, the modification required is solely to replace the VC-dimension term $d$ in the sample size with the cardinality term $\log |\mathcal{R}|$.
That is, the Bernoulli and reservoir sampling algorithms with sample size $\Omega(\log |\mathcal{R}|/\varepsilon^2)$ output a representative sample of the stream with good probability, even in the presence of an adaptive adversary. This nearly matches the bound imposed by the attack.

\end{abstract}

\newpage

\setcounter{page}{1}

\section{Introduction}

Random sampling is a simple, generic, and universal method to deal with massive 
amounts of data across all scientific disciplines. It has wide-ranging 
applications in  statistics, 
databases, 
networking, data mining, approximation 
algorithms, randomized algorithms, machine learning, and other fields (see 
e.g.,~\cite{CranorJSS03,JohnsonMR05,JermainePA04,CohenDKLT11,CormodeG05,CormodeMY11}
and \cite[Chapter 4]{Chazelle01}).
Perhaps the central reason for its wide applicability is the fact that it 
(provably, and with high probability) suffices to take only a small number of 
random samples from a large dataset in order to ``represent'' the dataset 
truthfully (the precise geometric meaning is explained later).
Thus, instead of performing costly and sometimes infeasible computations on 
the full dataset, one can sample a small 
yet ``representative'' subset of a data, perform the required analysis on 
this small subset, and extrapolate (approximate) conclusions from the 
small subset to the entire dataset.

The analysis of sampling algorithms has mostly been studied in the 
\emph{non-adaptive} (or \emph{static}) setting, where the data is fixed in advance, and then the sampling procedure runs on the fixed data. However, it is not always realistic to assume that the data does not change during the sampling procedure, as described in~\cite{MironovNS11,GilbertHRSW12,GilbertHSWW12,HardtW13,NaorY15}. 
In this work, we study the robustness of sampling in an \emph{adaptive}
adversarial environment.

\paragraph{\bf The adversarial environment.}
In high-level, the model is a two-player game between a randomized streaming 
algorithm, called $\Sampler$, and an adaptive player, $\Adversary$. 
In each round,
\begin{enumerate}
\item $\Adversary$ first submits an element to $\Sampler$. The choice of the element can depend, possibly in a probabilistic manner, on all elements submitted by $\Adversary$ up to this point, as well as all information that $\Adversary$ observed from $\Sampler$ up to this point.
\item Next, $\Sampler$ probabilistically updates its internal state, i.e., the 
sample that it currently maintains. An update step usually involves an 
insertion of the newly received element to the sample with some probability, 
and sometimes deletion of old elements from the sample.

\item Finally, $\Adversary$ is allowed to observe the 
current (updated) state of $\Sampler$, before proceeding to the next round.
\end{enumerate}

$\Adversary$'s goal is to make the sample as {\em unrepresentative} as 
possible, causing $\Sampler$ to come with false conclusions about the data 
stream. The 
game is formally described in Section \ref{subsec:adversarial_model_rules}.

Adversarial scenarios are common and arise in different settings. An adversary 
uses {\em adversarial examples} to fool a trained machine learning 
model~\cite{SzegedyZSBEGF13,MontasserHS19}; In the field of online learning \cite{Hazan2013}, adversaries are typically adaptive \cite{ShamirSzlak2017, Lykouris2018}. An online store suggests 
recommended items based on a sample of 
previous purchases, which in turn influences future sales 
\cite{Shalev-Shwartz12,GilbertHRSW12}. A 
network device routes traffic according to statistics pulled from a sampled 
substream of 
packets \cite{DuffieldLT05}, and an adversary that 
observes the network's traffic learns the device's routing choices might cause 
a denial-of-service attack by generating a small amount of adversarial traffic 
\cite{NaorY15}.
A high-frequency stock trading algorithm monitors a stream of stock orders 
places 
buy/sell requires based on statistics drawn from samples; A competitor might 
fool the sampling algorithm by observing its requests and modifying future 
stock orders accordingly.
An autonomous vehicle receives physical signals from its immediate environment (which might be adversarial \cite{Sitawarin2018}) and has to decide on a suitable course of action.

Even when there is no apparent adversary, the adaptive perspective is sometimes 
natural and required. 
For instance, adaptive data analysis \cite{Dwork2015Science, Woodworth18} aims 
to understand the challenges arising when data arrives online, such as data 
reuse, the implicit bias ``collected'' over time in scientific discovery, and 
the evolution of statistical hypotheses over time. In graph algorithms, 
\cite{ChuGPSSW18} observed that an adversarial analysis of dynamic spanners
would yield a simpler (and quantitively  better) alternative to their work.
 

In view of the importance of robustness against adaptive adversaries, and the fact that random sampling is very widely used in practice (including in streaming settings), we ask the following.

\begin{quotation}
\begin{center}
{\em 
Are sampling algorithms robust against adaptive adversaries?
}
\end{center}
\end{quotation}

\paragraph{\bf Bernoulli and reservoir sampling.}
We mainly focus on two of the most basic and well-known sampling algorithms: 
Bernoulli sampling 
and reservoir sampling. 
The Bernoulli sampling algorithm with parameter $p \in [0,1]$ runs as follows: 
whenever it receives a stream element $x_i$, the algorithm stores the element 
with probability $p$.
For a stream of length $n$ the sample size is expected to 
be $np$; and furthermore, it is well-concentrated around this value.
We denote this algorithm by $\UniformSample$. 

The classical reservoir sampling algorithm \cite{Vitter85} (see also 
\cite[Section 3.4.2]{Knuth1997} and a formal description in 
\Cref{sec:model}) with parameter $k \in [n]$ maintains a uniform 
sample of fixed size $k$, acting as follows. The first $k$ elements it 
receives, $x_1, 
\ldots, x_k$, are simply added to the memory with probability one. When the 
algorithm receives its $\ith{i}$ element $x_i$, where $i > k$, it stores it 
with probability $k/i$, by overriding a uniformly random element 
from the memory (so the 
memory size is kept fixed to $k$).
We henceforth denote this algorithm by $\ReservoirSample$.

\paragraph{\bf Attacking sampling algorithms.}
To answer the question above of whether sampling algorithms are robust against 
adversarially chosen streams, we must first define a notion of a 
representative sample, as several notions might be appropriate. However, we 
begin the discussion with an example showing how to attack the Bernoulli (and 
reservoir) sampling algorithm with respect to merely any definition of 
``representative''.

Consider a setting where the stream 
consists of $n$ points $x_1, \ldots, x_n$ in the 
one-dimensional range of real numbers $[0,1]$. $\UniformSample$ receives 
these points and samples each one 
independently with probability $p<1$. 
One can observe that, in the static setting and for sufficiently large $p$, the 
sampled set will be a good representation of the entire $n$ points for various 
definitions of the term ``representation''. For example, the median of 
the stream will be $\varepsilon$-close\footnote{The term ``close'' here 
means that the median of the sampled set 
will be an element whose order among the elements of the full stream, when the elements are sorted by value from smallest to largest, is within the range $(1 \pm \epsilon)n/2$, with high probability where the 
parameter $\epsilon$ depends on the probability $p$.} to the median of 
the sampled elements with high probability, as long as 
$p=\frac{c}{\varepsilon^2n}$ for some constant $c > 0$ (this also holds for 
any other quantile).

Consider the following {\em adaptive} adversary which will demonstrate 
the difference of the adaptive setting. $\Adversary$ keeps a ``working range'' at any point during the game, starting with the full range $[0,1]$. In the first round, $\Adversary$ chooses the number $x_1=0.5$ as the first element in the stream. If $x_1$ is 
sampled, then $\Adversary$ moves to the range $[0.5,1]$, and otherwise, to the 
range $[0,0.5]$. Next, $\Adversary$ submits $x_2$ as the middle of the current range. This 
continues for $n$ steps; Formally, $\Adversary$'s strategy is as 
follows. Set $a_1=0$ and $b_1=1$. In round $i$, where $i$ runs from $1$ to $n$, $\Adversary$ submits $x_i=\frac{a_i+b_i}{2}$ to $\UniformSample$; If 
$x_i$ is sampled then $\Adversary$ sets $a_{i+1}=x_i, b_{i+1} = b_i$, and otherwise, it sets $a_{i+1} = a_i, b_{i+1}=x_i$. The final stream is 
$x_1,\ldots,x_n$.

Note that at any point throughout  the process, $\Adversary$ always 
submits an element that is
larger than all elements in the current sampled set, and also smaller 
than all 
the non-sampled elements of the stream. Therefore, the end result is that 
after this 
process is over, with probability 1, the $k$ sampled elements 
are precisely the {\em smallest} $k$ elements in the stream. Of course, the median 
of the sampled set is far from the median of the stream as such a 
subset is very 
{\em unrepresentative} of the data. Actually, one might consider it as 
the ``most unrepresentative'' subset of the data.

The exact same attack on $\UniformSample$ works almost as effectively against 
$\ReservoirSample$. In this case, the attack will cause all of the $k$ sampled 
elements at the end of the process to lie among the first $O(k \ln n)$ 
elements with high probability. For more details, see Section \ref{sec:attack}.

\paragraph{\bf The good news.}
This attack joins a line of attacks in the adversarial model.
Lipton and Naughton \cite{LiptonN93} showed that an adversary that can measure 
the
time of operations in a dictionary can use this information to increase the 
probability  of a collision and as a result, significantly decrease the 
performance of the hashtable.
Hardt and 
Woodruff \cite{HardtW13} showed that linear sketches are inherently 
non-robust and cannot be used to compute the Euclidean norm of its input (where in 
the static setting they are used mainly for this reason). Naor and Yogev 
\cite{NaorY15} showed that Bloom filters are susceptible to attacks by an 
adaptive stream of queries if the adversary is computationally unbounded 
and they also constructed a robust Bloom filter against computationally 
bounded adversaries.

In our case, we note that the given attack might categorize it as 
``theoretical'' only. In practice, it is unrealistic to assume that the universe from which $\Adversary$ can pick elements is an infinite set; how would the attack look, then, if the universe is the discrete set $[N] = \{1, \ldots, N\}$?
$\Adversary$ 
splits the range $[0,1]$ to half for $n$ times, meaning that the precision of 
the elements required is exponential; The analogous attack in the discrete setting requires $N$ to be exponentially large with respect to the stream size $n$. 
Such a universe size is large and ``unrealistic'': for $\Sampler$ to memorize even a single element  requires memory size that is linear in $n$, whilst sampling and streaming algorithms usually aim to use an amount sublinear in $n$ of memory.

Thus, the question 
remains whether there exist attacks that can be performed on elements using 
substantially less 
precision, that is, on a significantly smaller size of discrete universe. In 
this work, we bring good news to both the Bernoulli and reservoir sampling 
algorithms by answering this question \emph{negatively}. We show that both sampling 
algorithms, with the right parameters, will output a representative sample with good probability regardless of $\Adversary$'s strategy, thus exhibiting robustness for these algorithms in adversarial settings.

We note that any {\em deterministic} algorithm that works in the static setting 
is inherently robust in the adversarial adaptive setting as well. However, in many cases, deterministic 
algorithms with small memory simply do not exist, or they are complicated and
tailored for a specific task. Here, we enjoy the simplicity of a generic 
randomized sampling algorithm combined with the robust guarantees of our 
framework.

\paragraph{\bf What is a representative sample?}
Perhaps the most standard and well-known notion of being representative is that 
of an $\varepsilon$-approximation, first suggested by Vapnik and Chervonenkis~\cite{VC1971} (see also~\cite{MustafaVaradarajan:2017}), which originated as a natural notion of discrepancy~\cite{Chazelle01} in the geometric 
literature. It is closely related to the celebrated notion of VC-dimension \cite{VC1971, Sauer1972, Shelah1972}, and captures many quantitative properties 
that are desired in a random subset.
Let $X=(x_1,\ldots,x_n)$ be a sequence of elements from the universe $U$ (repetitions are allowed) and let $R
\subseteq U$. The \emph{density} of $R$ in $X$ is the 
fraction of elements in $X$ that are also in $R$ (i.e., $d_R(X) = \Pr_{i \in 
[n]}[x_i \in R]$).

A \emph{set system} is simply a pair $(U, \mathcal{R})$ where $\mathcal{R} \subseteq 2^U$ is a collection of subsets.
A non-empty subsequence $S$ of $X$ is an \emph{$\varepsilon$-approximation} of $X$ with 
respect to the set system $(U, \mathcal{R})$ if it preserves densities (up to 
an $\varepsilon$ factor) for all 
subsets $R \in \mathcal{R}$.

\begin{definition}[$\varepsilon$-approximation]
We say that a (non-empty) sample $S$ is an $\varepsilon$-{\em approximation} of $X$ with 
respect to $\mathcal{R}$ if for any subset $R \in \mathcal{R}$ it holds that
$
\left| d_{R}(X) - d_{R}(S) \right| \le \varepsilon.
$
\end{definition}
If the universe $U$ is well-ordered, it is natural to take $\mathcal{R}$ as the 
collection of all consecutive intervals in $U$, that is, $\mathcal{R} = \{[a,b] 
: a \leq b \in U\}$ (including all singletons $[a,a]$). 
With this set system in hand, $\varepsilon$-approximation is a natural form of ``good 
representation'' in the streaming setting, pointed out by its deep connection to 
multiple classical problems in the streaming literature, like approximate median, 
and more generally, quantile estimation \cite{Manku1999, Greenwald2001, 
Wang2013, Greenwald2016,  KarninLangLiberty2016} and range searching 
\cite{BagchiCEG07}.
In particular, if $S$ is an 
$\varepsilon$-approximation of $X$ w.r.t. $(U, \mathcal{R})$, then any $q$-quantile 
of $S$ is $\varepsilon$-close to the $q$-quantile of $X$; this holds 
simultaneously for all quantiles (see \Cref{sec:applications}).


\subsection{Our Results}
Fix a set system $(U, \mathcal{R})$ over the universe $U$.
A sampling algorithm is called \emph{$(\varepsilon, \delta)$-robust} if for 
any (even computationally unbounded) strategy of $\Adversary$, the output sample $S$ 
is an $\varepsilon$-approximation of the whole stream $X$ with respect to $(U, \mathcal{R})$, with probability at least 
$1-\delta$.

Our main result is an upper bound (``good news'') 
on the 
$(\varepsilon, \delta)$-robustness of Bernoulli and reservoir sampling, later 
to be complemented them with near-matching lower bounds. 
\begin{theorem}
\label{thm:UB}
For any $0 < \varepsilon, \delta < 1$, 
set system $(U, \mathcal{R})$, and stream length $n$, the following holds.
\begin{itemize}
	\item $\UniformSample$ with parameter 
	$p \geq 10\cdot\frac{\ln|\mathcal{R}| + \ln 
		(4/\delta)}{\varepsilon^2n}$ is 
	$(\varepsilon,\delta)$-robust.
	\item $\ReservoirSample$ with parameter \ 
	$k \geq 2\cdot\frac{\ln|\mathcal{R}| + \ln 
		(2/\delta)}{\varepsilon^2}$ is 
	$(\varepsilon,\delta)$-robust.
\end{itemize}
\end{theorem}
The proof appears in Section \ref{sec:main-technical}.
As the total number of elements sampled by $\UniformSample$ is well-concentrated 
around $np$, the above theorem implies that a sample of total size (at least) 
$\Theta(\frac{\ln|\mathcal{R}| + \ln 
	\frac{1}{\delta}}{\varepsilon^2})$, obtained by any of the algorithms, 
	$\UniformSample$ or 
	$\ReservoirSample$, is an $\varepsilon$-approximation with probability 
	$1-\delta$.

This should be compared with the static setting, where the 
same result is known as long as $p\ge c \cdot \frac{d + \ln 
\frac{1}{\delta}}{\varepsilon^2n}$ for $\UniformSample$, and $k \ge c \cdot 
\frac{d + \ln \frac{1}{\delta}}{\varepsilon^2}$ for $\ReservoirSample$, where 
$d$ is the VC-dimension of $(U, \mathcal{R})$ and $c>0$ is a constant 
\cite{VC1971,talagrand1994sharper,LiLS01} (see also 
\cite{MustafaVaradarajan:2017}).

As you can see, to make the static sampling algorithm robust in the adaptive setting 
one solely needs to modify the sample size by replacing the VC-dimension term 
$d$ with the cardinality 
dimension $\ln|\mathcal{R}|$ (and update the multiplicative constant). Below, 
in our lower bounds, we 
show that this increase in the sample size is inherent, and not a byproduct of 
our 
analysis.

\paragraph{\bf Lower Bounds.}
We next show that being adaptively robust comes at a 
price. That is, the dependence on the cardinality dimension, as opposed to the 
VC dimension, is necessary. By an improved version of the attack described in 
the introduction, we show the following:
\begin{theorem}\label{thm:attack}
There exists a constant $c>0$ and a set system $(U,\mathcal{R})$ with 
VC-dimension 1, where such that for 
any $0 < \varepsilon, \delta < 1/2$:
\begin{enumerate}
	\item The $\UniformSample$ algorithm with parameter 
	$p < c \cdot \frac{\ln |\mathcal{R}|}{n \ln n}$ is \textbf{not} 
	$(\varepsilon,\delta)$-robust.
	\item The $\ReservoirSample$ algorithm with parameter 
	$k <c \cdot \frac{\ln |\mathcal{R}|}{\ln n}$ is \textbf{not} 
	$(\varepsilon,\delta)$-robust.
\end{enumerate}
Moreover, for any $n^{6 \ln n} \leq N \leq 2^{n/2}$, there exists $(U, 
\mathcal{R})$ as above where $|\mathcal{R}|=|U|=N$.
\end{theorem}
\noindent
The proof can be found in Section \ref{sec:attack}.
\paragraph{\bf Continuous robustness.}
The condition of $(\varepsilon, \delta)$-robustness requires that the sample will be $\varepsilon$-representative  of the stream \emph{in the end of the process}. What if we wish the sample to be representative of the stream \emph{at any point} throughout the stream?
Formally, we say that a 
sampling algorithm is \emph{$(\varepsilon, \delta)$-continuously robust} if, 
with probability 
at least $1-\delta$, at \emph{any} point $i \in [n]$ the sampled set $S_i$ is an 
$\varepsilon$-approximation of the first $i$ elements of the stream, i.e., of $X_i = (x_1, \ldots, x_i)$.
The next theorem shows that continuous robustness of $\ReservoirSample$ can be obtained with just a small overhead compared to ``standard'' robustness. (For $\UniformSample$ one cannot hope for such a result to be true, at least for the above definition of continuous robustness.)
\begin{theorem}
	\label{thm:continuous}
	There exists $c>0$, such that for any $0 < \varepsilon, \delta < 1/2$, set 
	system $(U, \mathcal{R})$, and stream length $n$, $\ReservoirSample$ with 
	parameter 
	$k \ge c \cdot \frac{\ln |\mathcal{R}| + \ln 1 / \delta + \ln 1 / \varepsilon + \ln \ln n}{\varepsilon^2}$ 
	is 
	$(\varepsilon,\delta)$-continuously robust. 
	
	Moreover, if only continuous robustness against a static adversary is 
	desired, then the $\ln |\mathcal{R}|$ term can be replaced with the 
	VC-dimension of $(U, \mathcal{R})$.
\end{theorem}
We are not aware of a previous analysis of continuous robustness, even in the static 
setting. The proof, appearing in \Cref{sec:continuous}, follows by applying Theorem \ref{thm:UB} (or its static 
analogue) in carefully picked ``checkpoints'' $k = i_1 \leq i_2 \leq \ldots 
\leq i_t = n$ along the stream, where $t = O(\varepsilon^{-1} \ln n)$. It shows 
that if the sample $S_i$ is representative of the stream $X_i$ in any of the 
points $i=i_1, \ldots, i_{t-1}$, then with high probability, the sample is also 
representative in any other point along the stream. 
(We remark that a similar statement with weaker dependence on $n$ can be 
obtained from Theorem \ref{thm:UB} by a straightforward union bound.) The proof 
can be found in \Cref{sec:continuous}.

\paragraph{\bf Comparison to deterministic sampling algorithms.}
Our results show that sampling algorithms provide an $\varepsilon$-approximation 
in the adversarial model. One advantage of using the notion of 
$\varepsilon$-approximation is its wide array of applications, where for each such task we get a
streaming algorithm in the adversarial model as described in the following 
subsection. We stress that for any specific task a {\em deterministic} 
algorithm that works in the static setting will also automatically be 
robust in the adversarial setting. However, deterministic 
algorithms tend to be more complicated, and in some cases they require larger memory.
Here, we focus on showing that the most simple and 
generic sampling algorithms ``as is'' are robust in our adaptive model and 
yield a representative sample of the data
that can be used for many different applications. 

The best known {\em deterministic} algorithm for computing an 
$\varepsilon$-approximating sample in the streaming model is that of Bagchi et 
al.~\cite{BagchiCEG07}. The sample size they obtain is $O(\varepsilon^{-2} \ln 
1/\varepsilon)$; the working space of their algorithm and the processing time per element are of the form 
$\varepsilon^{-2d-O(1)} (\ln n)^{O(d)}$, where $d$ is the \emph{scaffold 
dimension}\footnote{The scaffold dimension is a variant of the VC-dimension 
equal to $\lceil\ln|\mathcal{R}| / \ln|U|\rceil$.} of the set system. 
The exact bounds are rather intricate, see Corollary 4.2 in~\cite{BagchiCEG07}. 
While the space requirement of their approach does not have a dependence on
$\ln |\mathcal{R}|$,
its dependence on $\varepsilon$ and $\ln n$ is generally worse 
than ours, making their bounds somewhat incomparable to ours. 
Finally, we note that there exist more efficient methods to generate an $\varepsilon$-approximation in some special cases, e.g., when the set system constitutes of rectangles or halfspaces \cite{Suri2004}. 

\subsection{Applications of Our Results}\label{sec:applications}
We next describe several representative applications and usages of 
$\varepsilon$-approximations (see also \cite{BagchiCEG07} for more applications in 
the area of robust statistics). For some of these applications, there exist 
deterministic algorithms known to require less memory than the simple random 
sampling models discuss in this paper. However, one area where our generic 
random sampling approach shines compared to deterministic approaches is the 
\emph{query complexity} or \emph{running time} (under a suitable computational 
model). Indeed, while deterministic algorithms must inherently query \emph{all} 
elements in the stream in order to run correctly, our random sampling 
methods query just a small sublinear portion of the elements in the stream. 

Consequently, to the best of our knowledge, Bernoulli and reservoir sampling are the first two methods known to compute an 
$\varepsilon$-approximation (and as a byproduct, solve the tasks described in this subsection) in adversarial situations where it is unrealistic 
or too costly to query all elements in the stream. The last part of this subsection 
exhibits an example of one such situation.

\paragraph{\bf Quantile approximation.}
As was previously mentioned, $\varepsilon$-approximations have a deep 
connection 
to approximate median (and more generally, quantile estimation).
Assume the universe $U$ is well-ordered.
We say that a streaming algorithm is an $(\varepsilon,\delta)$-\emph{robust 
quantile 
sketch} if, in our {\em adversarial model}, it provides a sample that allows to 
approximate the rank\footnote{The \emph{rank} of an element $x_i$ in a stream $x_1, \ldots, x_n$ is the total amount of elements $x_j$ in the stream so that $x_j \leq x_i$.} of any 
element in the stream up to additive error $\varepsilon n$ with probability at least $1-\delta$.
Observe that this is achieved with an $\varepsilon$-approximation with 
respect to the set system $(U,\mathcal{R})$ where $\mathcal{R}=\{[1,b] : b \in 
U\}$. For example, set $b$ to be the median 
of the stream. Since the density of the range $[1,b]$ is preserved in the 
sample, we know that the median of the sample will be $\varepsilon$-close to 
the median of the stream. This works for any other quantile simultaneously.
The 
sample size 
is $\Theta(\frac{\ln|U| + \ln 
	(1/\delta)}{\varepsilon^2})$.
\begin{corollary}                                                               
For any $0 < \varepsilon, \delta < 1$, well-ordered universe $U$, and stream 
length $n$, $\UniformSample$ with parameter $p \geq 10\cdot\frac{\ln|U| + \ln 
(4/\delta)}{\varepsilon^2n}$ is an $(\varepsilon,\delta)$-robust quantile 
sketch. The same holds for the $\ReservoirSample$ algorithm with parameter $k 
\geq 
2\cdot\frac{\ln|U| + \ln (2/\delta)}{\varepsilon^2}$.
\end{corollary}
A corollary in the same spirit regarding \emph{continuously} robust quantile sketches can 
be derived from Theorem \ref{thm:continuous}. 

\paragraph{\bf Range queries.}
Suppose that the universe is of the form $U=[m]^d$ for some parameters $m$ and 
$d$. One basic problem is that of {\em range queries}: one is given a set of 
ranges $\mathcal{R}$ and each query consists of a range $R \in \mathcal{R}$ 
where the desired answer is the number of points in the stream that are in this 
range. Popular choices of such ranges are axis-aligned or rotated boxes, spherical ranges 
and simplicial ranges. An $\varepsilon$-approximation allows us to answer such 
range queries up to an additive error of
$\varepsilon n$. Suppose the sampled set is $S$, then an answer is given by 
computing $d_{R}(S) \cdot n/|S|$. For example, when $\mathcal{R}$ consists of 
all 
axis-parallel boxes, $\ln |\mathcal{R}| = O(d\ln m)$ and thus the 
sample size required to answer range queries that are robust against 
adversarial streams is $|S| = O\left(\frac{d\ln|m| + \ln 
(1/\delta)}{\varepsilon^2}\right)$; for rotated boxes, one should replace $d$ 
with $d^2$ in this expression. See \cite{BagchiCEG07} for 
more details on the connection between $\varepsilon$-approximations and 
range queries.

\paragraph{\bf Center points.}
Our result is also useful for computing $\beta$-center points. A point $c$ in 
the stream is a $\beta$-center point if every closed halfspace containing $c$ 
in fact contains at least $\beta n$ points of the stream. In \cite[Lemma 
6.1]{ClarksonEMST96} it has been shown that an $\varepsilon$-approximation (with respect to half-spaces) can 
be used to get a $\beta$-center point for suitable choices of the parameters. For example, 
setting $\varepsilon=\beta/5$ we get that a $6\beta/5$-center of the sample 
$S$ is a $\beta$-center of the stream $X$. Thus, we can compute a 
$\beta$-center of a stream in the adversarial model. See also 
\cite{BagchiCEG07}.

\paragraph{\bf Heavy hitters.}
Finding those elements that appear many times in a stream is a fundamental problem in data mining, with a myriad of practical applications.
In the \emph{heavy hitters} problem, there is a threshold  $\alpha$ and an error parameter $\varepsilon$. The 
goal 
is to output a list of elements such that if an element $x$ appears more than 
$\alpha n$ 
times in the stream (i.e., $d_x(X) \ge \alpha$) it must be included in the 
list, and if an element appears 
less than $(\alpha-\varepsilon)n$ times in the stream (i.e., $d_x(X) \le \alpha 
- \varepsilon$ it cannot be included in the list.

Our results yield a simple and efficient heavy hitters streaming algorithm in 
the 
adversarial model. For any universe $U$ let
$\mathcal{R}=\{\{a\} : a \in U\}$ be the set of all singletons. Now, pick 
$\varepsilon' = \varepsilon/3$ and use either Bernoulli or reservoir sampling 
to compute an $\varepsilon'$-approximation $S$ of the stream $X$, outputting 
all elements $x \in S$ with $d_{\{x\}}(S) 
\ge \alpha - \varepsilon'$. Indeed, if $d_a(X) \ge \alpha$ then $d_x(S) \ge 
\alpha - \varepsilon'$. On the 
other hand, if $d_x(X) \le \alpha - \varepsilon$ then $d_x(S) \le 
\alpha-\varepsilon + \varepsilon' < \alpha-\varepsilon'$.
\begin{corollary}                                                               
There exists $c>0$ such that for any $0 < \varepsilon, \delta < 1/2$, 
universe $U$, and stream length $n$, $\UniformSample$ with 
parameter $p \geq c\cdot\frac{\ln|U| + \ln 
	(1/\delta)}{\varepsilon^2n}$ solves the heavy hitters problem with error 
	$\varepsilon$ in the adversarial model. The same holds for 
	$\ReservoirSample$ with parameter $k \geq c \cdot\frac{\ln|U| + 
	\ln 
	(1/\delta)}{\varepsilon^2}$.
\end{corollary}

\paragraph{\bf Clustering.}
The task of partitioning data elements into separate groups, where the elements 
in each group are ``similar'' and elements in different groups are 
``dissimilar'' is fundamental and useful for numerous applications across 
computer science. There has been lots of interest on clustering in a streaming 
setting, see e.g.~\cite{Ghesmoune2016} for a survey on recent results. Our 
results suggest a generic framework to accelerate clustering algorithms in the 
adversarial model: Instead of running 
clustering on the full data, one can simply sample the data to obtain (with 
high probability, even against an adversary) an $\varepsilon$-approximation of 
it, run the clustering algorithm on the sample, and then extrapolate the 
results to the full dataset.

\paragraph{\bf Sampling in modern data-processing systems.}
It is very common to use random sampling (sometimes ``in disguise'') in 
modern data-intensive systems that operate on streaming data, arriving in an 
online manner.
As an illustrative example, consider the following \emph{distributed database} \cite{Ozsu2011} setting. 
Suppose that a database system must receive and process a huge amount of queries per second. It is unrealistic for a single server to handle all the queries, and hence, for load balancing purposes, each incoming query is 
randomly assigned to one of $K$ query-processing servers. 
Seeing that the set of queries that each such server receives is 
essentially a Bernoulli random sample (with parameter $p = 1/K$) of the full 
stream, one hopes that the portion of the stream sampled by each of these 
servers would truthfully represent the whole data stream (e.g., for query optimization purposes), even if the 
stream changes with time (either unintentionally or by a malicious adversary). 
Such ``representation guarantees'' are also desirable in distributed machine 
learning systems \cite{goyal2017imagenet1hr, smith2017decay}, where each processing unit 
learns a model according to the portion of the data it received, and the models 
are then aggregated, with the hope that each of the units processed ``similar'' 
data.

In general, modern data-intensive systems like those described above become 
more and more complicated with time, consisting of a large number of different 
components. Making these systems \emph{robust} against environmental changes in the data, 
let alone {\em adversarial} changes, is one of the greatest challenges in 
modern computer science. 
From our perspective, the following question naturally emerges:
\begin{center}
	\emph{Is random sampling a risk in modern data processing systems?}
\end{center}
Fortunately, our results indicate that the answer to this question is largely \emph{negative}. 
Our upper bounds, Theorems \ref{thm:UB} and \ref{thm:continuous}, show that a sufficiently large sample suffices to circumvent adversarial changes of the environment.

\subsection{Related Work}

\paragraph{\bf Online learning.}

One related field to our work is {\em online learning}, which was 
introduced for settings where the data is given in a sequential online manner 
or where it is necessary for the learning algorithm to adapt to changes in the 
data. Examples include stock price predictions, ad click prediction, and more 
(see \cite{Shalev-Shwartz12} for an overview and more examples). 

Similar to our model, online learning is viewed as a repeated game between a 
learning algorithm (or a 
predictor) and the environment (i.e., the adversary). It considers $n$ rounds 
where in each round the environment submits an instance $x_i$, the learning 
algorithm then makes a prediction for $x_i$, the environment, in turn, chooses 
a 
loss for this prediction and sends it as feedback to the algorithm. The goal in this model is usually
to minimize regret (the sum of losses) compared to the best fixed 
prediction in hindsight. This is the typical setting (e.g., 
\cite{HazanAK07,SrebroST10}), however, many different variants exist (e.g., 
\cite{DanielyGS15,ZhangLZ18}).


\paragraph{\bf PAC learning.}
In the PAC-learning framework \cite{Valiant1984}, the learner algorithm 
receives samples generated from an unknown distribution and must choose a 
hypothesis function from a family of hypotheses that best predicts the data 
with respect to the given distribution. It is known that the number of samples 
required for a class to be learnable in this model depends on the VC-dimension 
of the class.

A recent work of Cullina et al.~\cite{Cullina2018} investigates 
the effect of evasion adversaries on the PAC-learning framework, coining the 
term of \emph{adversarial VC-dimension} for the parameter governing the sample 
complexity. Despite the name similarity, their context is seemingly unrelated 
to ours (in particular, it is not a streaming setting), and correspondingly, 
their notion of adversarial VC-dimension does not seem to relate to our work.

\paragraph{\bf Adversarial examples in deep learning.}

A very popular line of research in modern deep learning proposes methods to attack neural networks, and countermeasures to these attacks. In such a 
setting, an adversary performs adaptive queries to the learned model in order to fool the model via a malicious input. The learning 
algorithms usually have an underlying assumption that the training and test 
data are generated from the same statistical distribution. However, in 
practice, the presence of an adaptive adversary violates this assumption.
There are many devastating examples of attacks on learning models 
\cite{SzegedyZSBEGF13,BiggioCMNSLGR13,PapernotMGJCS17,BiggioR18,MontasserHS19} 
and 
we stress that currently, the understanding of techniques to 
defend against such adversaries is rather limited
\cite{GoodfellowMP18,McCoyd018,MahloujifarM19,MontasserHS19}. 

\paragraph{\bf Maintaining random samples.}
Reservoir sampling is a simple and elegant algorithm for maintaining a random 
sample of a stream \cite{Vitter85}, and since its proposal, many flavors have 
been introduced.
Chung, Tirthapura, Woodruff \cite{Chung2016ASM} generalized reservoir sampling 
to the setting of multiple distributed streams, which need to coordinate in 
order to continuously respond to queries over the union of all streams observed 
so far (see also Cormode et al.~\cite{Cormode2012ContinuousSF}). 
Another variant is weighted reservoir sampling where
the probability of sampling an element is proportional to a weight associated 
with the element in the stream \cite{EfraimidisS06,BravermanOV15}. A 
distributed version as above was recently considered for the weighted case as well 
\cite{Jayaram2019}.

\subsection{Paper Organization}
\Cref{sec:model} contains an overview of our adversarial model and a 
more precise and detailed definition than the one given in the introduction. In 
\Cref{sec:prelims} we mention several concentration inequalities required for our analysis. 
In 
\Cref{sec:main-technical} we present and prove our main technical Lemma,
from which we derive Theorem \ref{thm:UB}. This includes analysis of both 
$\UniformSample$ and $\ReservoirSample$. In \Cref{sec:attack} we 
present our ``attack'', i.e., our lower bound showing the tightness of our 
result. Finally, in \Cref{sec:continuous}, we prove our upper bounds in the 
continuous setting.


\section{The Adversarial Model for Sampling}
\label{sec:model}
In this section, we formally define the online adversarial model discussed in 
this paper. Roughly speaking, we say that $\Sampler$ is an
$(\varepsilon,\delta)$-robust sampling algorithm for a set system 
$(U,\mathcal{R})$ if for any adversary
choosing an adaptive stream of elements $X = (x_1, \ldots, x_n)$, the final 
state of the 
sampling algorithm $\sigma_n$ is an $\varepsilon$-approximation of the stream 
with probability $1-\delta$.
This is formulated using a game, $\AdaptiveGame$, between 
two players, $\Sampler$ and $\Adversary$.

\paragraph{\bf Rules of the game:}
\label{subsec:adversarial_model_rules}
\begin{enumerate}
	\item $\Sampler$ is a streaming algorithm, which gets a sequence of $n$ elements one by 
	one $x_1,\ldots,x_n$ in an online manner (the sampling algorithms we discuss in this paper do not need to know $n$ in advance).
	Upon receiving an element $x_i$, 
	$\Sampler$ can perform an arbitrary computation (the running time can be
	unbounded) and update a local state $\sigma$. We denote the local state after $i$ steps by $\sigma_i$, and write $\sigma_i 
	\gets \Sampler(\sigma_{i-1},x_i)$.
	\item The stream is chosen adaptively by $\Adversary$: 
	a probabilistic (unbounded) player that, given all previously sent elements $x_1, \ldots, x_{i-1}$ and 
	the current state $\sigma_{i-1}$, chooses the next element $x_i$ to submit. The \emph{strategy} that Adversary employs along the way, that is, the probability distribution over the choice of $x_i$ given any possible set of values $x_1, \ldots, x_{i-1}$ and $\sigma_{i-1}$, is fixed in advance. %
	The underlying (finite or infinite) set from which $\Adversary$ is allowed 
	to choose elements during the game is called the \emph{universe}, and 
	denoted by $U$. We assume that $U$ does not change along the game.
	\item Once all $n$ rounds of the game have ended, $\Sampler$ outputs $\sigma_n$. For the sampling algorithms discussed in this paper, $S := \sigma_n$ is a subsequence 
	of the stream $X = (x_1, \ldots, x_n)$.
	$S$ is usually called the \emph{sample} obtained by $\Sampler$ in 
	the game.
\end{enumerate}
For an illustration on the rules of the game see \Cref{fig:adaptive-game}.

\protocol
{The game $\AdaptiveGame$}
{The definition of the game $\AdaptiveGame_{\varepsilon}$ between a streaming algorithm $\Sampler$ and $\Adversary$. Here the adversary chooses the next element to the stream while given the state (memory) of the streaming algorithm thus-far. In the beginning of the game, $\Adversary$ receives the parameters $n, \varepsilon, (U, \mathcal{R})$ and knows exactly which sampling algorithm is employed by $\Sampler$.}
{fig:adaptive-game}
{
	\textbf{Parameters}: $n$, $\varepsilon$, $(U,\mathcal{R})$.
	\begin{enumerate}
		\item Set $\sigma_0 = \bot$.
		\item For $i=1 \ldots n$ do:
		\begin{enumerate}
			\item $\Adversary(\sigma_{i-1},x_1,\ldots,x_{i-1})$ submits the query $x_i$.
			\item Set $\sigma_i \gets \Sampler(\sigma_{i-1},x_i)$.
		\end{enumerate}
		\item Let $S=\sigma_n$, and output $1$ if $S$ is an 
		$\varepsilon$-representative sample 
		of $X=x_1,\ldots,x_n$ with respect to $(U,\mathcal{R})$, and $0$ 
		otherwise.
	\end{enumerate}
}

Using the game defined above, we now describe what it means for a sampling algorithm to be (adversarially) \emph{robust}. 
\begin{definition}[Robust sampling algorithm]
We say that a sampling algorithm $\Sampler$ is $(\varepsilon,\delta)$-robust with respect to the set system $(U, \mathcal{R})$ and the stream length $n$ if 
for and any (even unbounded) strategy of $\Adversary$, it holds that
	\begin{align*}
	\Pr[\AdaptiveGame(\Sampler,\Adversary)=1] \ge 1-\delta
	\end{align*}
	The memory size used by $\Sampler$ is defined to be the maximal size of $\sigma$ throughout the process of $\AdaptiveGame$.
\end{definition}
A stronger requirement that one can impose on the sampling algorithm is to hold an 
$\varepsilon$-approximation of the stream at \emph{any} step during the game.
To handle this, we define a continuous variant of $\AdaptiveGame$ which we denote 
$\ContinuousAdaptiveGame$, presented in \Cref{fig:cont-adaptive-game}.
\protocol
{The game $\ContinuousAdaptiveGame$}
{The game corresponding to the continuous variant, $\ContinuousAdaptiveGame$ between a 
streaming algorithm $\Sampler$ and an $\Adversary$. Here, $\Sampler$ is 
required to hold an $\varepsilon$-approximating sampled set $S_i$ after each step.}
{fig:cont-adaptive-game}
{
	\textbf{Parameters}: $n$, $\varepsilon$, $(U,\mathcal{R})$.
	\begin{enumerate}
		\item Set $\sigma_0 = \bot$.
		\item For $i=1 \ldots n$ do:
		\begin{enumerate}
			\item $\Adversary((\sigma_{i-1},S_{i-1}),x_1,\ldots,x_{i-1})$ 
			submits the query $x_i$.
			\item Set $(\sigma_i,S_i) \gets \Sampler(\sigma_{i-1},x_i)$.
			\item If $S_i$ is not an $\varepsilon$-approximation of 
			$X_i=x_1,\ldots,x_i$ with respect to $(U,\mathcal{R})$ then output 
			0 and halt.
		\end{enumerate}
		\item Output 1.
	\end{enumerate}
}

For the sampling algorithms that we consider, the state at any time $\sigma_i$ is essentially equal to the sample $S_i$. In any case, the definition of the framework given in \Cref{fig:cont-adaptive-game} generally allows $\sigma_i$ to contain additional information, if needed.
A sampling algorithm is called \emph{$(\varepsilon, \delta)$-continuously 
robust} if the following holds  with probability at least $1-\delta$: for any strategy of $\Adversary$, and all $i \in [n]$, the sample 
$S_i$ is an $\varepsilon$-approximation of the stream at time $i$. 
\begin{definition}[Continuously robust sampling algorithm]
We say that a sampling algorithm $\Sampler$ is $(\varepsilon,\delta)$-continuously robust with respect to the set system $(U, \mathcal{R})$ and the stream length $n$ if 
for and any (even unbounded) strategy of $\Adversary$, it holds that
\begin{align*}
\Pr[\ContinuousAdaptiveGame(\Sampler,\Adversary)=1] \ge 1-\delta
\end{align*}
The memory size used by $\Sampler$ is defined to be the maximal size of $\sigma$ 
throughout the process of $\ContinuousAdaptiveGame$.
\end{definition}

\paragraph{\bf Reservoir sampling.}\label{fig:reservoir-sampling}
For completeness, we provide the pseudocode of the reservoir sampling algorithm 
\cite{Vitter85, Knuth1997}. Here, $k$ denotes the (fixed) memory size of the 
algorithm, $i$ denotes the current round number, and $x_i$ is the currently 
received element.

\vspace{0.1cm}
\noindent \underline{$\ReservoirSample(k, i, \sigma_{i-1}, x_i)$}:
\begin{enumerate}
	\item If $i < k$ then parse $\sigma_{i-1}=x_1,\ldots,x_{i-1}$ and output 
	$\sigma_i= x_1,\ldots,x_i$.
	\item Otherwise, parse $\sigma_{i-1}=s_1,\ldots,s_k$.
	\item With probability $k/i$ do:\\ choose $j \in 
	[k]$ uniformly at random and output 
	$\sigma_i = \allowbreak s_1,\ldots,\allowbreak 
	s_{j-1},x_i,s_{j+1},\ldots,s_k$.
	\item Otherwise, output $\sigma_i=\sigma_{i-1}$.
\end{enumerate}

\section{Technical Preliminaries}
\label{sec:prelims}
The logarithms in this paper are usually of base $e$, and denoted by $\ln$. The
exponential function $\exp{(x)}$ is $e^x$. For an integer $n\in \N$ we denote 
by 
$[n]$ the set 
$\{1,\ldots,n\}$.
We state some concentration inequalities, useful for our analysis in later 
sections.
We start with the well-known Chernoff's inequality for sums of independent random variables.
\begin{theorem}[Chernoff Bound \cite{Chernoff1952}; see Theorem 3.2 in \cite{Chung2006}]\label{lem:chernoff}
	Let $X_1,\ldots,X_m$ be independent
	random variables that take the value 
	1 with probability $p_i$ and 0 otherwise, $X=\sum_{i=1}^{m}X_i$, and 
	$\mu=\E{X}$. Then for any $0 < \delta < 1$,
	$$\Pr[X \leq (1-\delta)\mu] \le \exp\left(-\frac{\delta^2 \mu}{2}\right)$$
	and
	$$\Pr[X \geq (1+\delta)\mu] \le \exp\left(-\frac{\delta^2 \mu}{2 + 2\delta/3 }\right).$$
\end{theorem}

Our analysis of adversarial strategies crucially makes use of \emph{martingale inequalities}. We thus provide the definition of a martingale.

\begin{definition}
	A {\em martingale} is a sequence $X = (X_0,\ldots,X_m)$ of random variables with finite means, so 
	that for $0 \le i < m$, it holds that $\E{X_{i+1} \mid X_{0},\ldots,X_i}=X_i$.
\end{definition}

The most basic and well-known martingale inequality, Azuma's (or Hoeffding's) inequality, asserts that martingales with bounded differences $|X_{i+1} - X_{i}|$ are well-concentrated around their mean. For our purposes, this inequality does not suffice, and we need a generalized variant of it, due to McDiarmid \cite[Theorem 3.15]{McDiarmid1998}; see also Theorem 4.1 in \cite{Freedman1975}. 
The formulation that we shall use is given as Theorem 6.1 in the survey of Chung and Lu \cite{Chung2006}.

\begin{lemma}[See \cite{Chung2006}, Theorem 6.1]
	\label{lem:martingale-with-var}
	Let $X = (X_0, X_1, \ldots, X_n)$ be a martingale. Suppose further that for any $1 \leq i \leq n$, the variance satisfies $\text{Var}(X_i | X_0, \ldots, X_{i-1}) \leq \sigma_i^2$ for some values $\sigma_1, \ldots, \sigma_n \geq 0$, and there exists some $M \geq 0$ so that $|X_i - X_{i-1}| \leq M$ always holds. Then, for any $\lambda \geq 0$, we have
	$$
	\Pr(X - X_0 \geq \lambda) \leq \exp \left(-\frac{\lambda^2}{2 \sum_{i=1}^{n} (\sigma_i^2) + M \lambda / 3} \right).
	$$ 
	In particular,
	$$
	\Pr(|X - X_0| \geq \lambda) \leq 2\exp \left(-\frac{\lambda^2}{2 \sum_{i=1}^{n} (\sigma_i^2) + M \lambda / 3}\right).
	$$
\end{lemma}
Unlike Azuma's inequality, Lemma \ref{lem:martingale-with-var} is well-suited to deal with martingales where the maximum value $M$ of $|X_{i+1} - X_{i}|$ is large, but the maximum is rarely attained (making the variance much smaller than $M^2$). The martingales we investigate in this paper depict this behavior.

\section{Adaptive Robustness of Sampling: Main Technical Result}
\label{sec:main-technical}
In this section, we prove the main technical lemma underlying our upper bounds 
for Bernoulli sampling and reservoir sampling. 
The lemma asserts that for both sampling methods, and any given subset $R$ of the universe $U$, the fraction of elements 
from $R$ within the sample typically does not differ by much from the corresponding 
fraction among the whole stream.

\begin{lemma}
	\label{lem:main_step_upper_bound}
Fix $\varepsilon, \delta > 0$, a universe $U$ and a subset $R\subseteq U$, and let $X = (x_1, x_2, \ldots, x_n)$ be the sequence chosen by $\Adversary$ in $\AdaptiveGame_{\varepsilon}$ against either $\UniformSample$ or $\ReservoirSample$. 

\begin{enumerate}
	\item \label{item:uniform} For $\UniformSample$ with parameter $p \geq 10 \cdot \frac{\ln(4/ \delta)}{\varepsilon^2 n}$, we have $\Pr(|d_R(X) - d_R(S)| \geq \varepsilon) \leq \delta$.
	\item \label{item:reservoir} For $\ReservoirSample$ with memory size $k \geq 2 \cdot \frac{ \ln(2/\delta)} {\varepsilon^2}$, it holds that $\Pr(|d_R(X) - d_R(S)| \geq \varepsilon) \leq \delta$. 
\end{enumerate}
\end{lemma}
Both of these bounds are tight up to an absolute multiplicative constant, even for a static adversary (that has to submit all elements in advance); see Section \ref{sec:continuous} for more details.

The proof of Theorem \ref{thm:UB} follows immediately from Lemma \ref{lem:main_step_upper_bound}, and is given below. The proof of Theorem \ref{thm:continuous} requires slightly more effort, and is given in Section \ref{sec:continuous}.
\begin{proof}[Proof of Theorem \ref{thm:UB}]
Let $(U, \mathcal{R})$, $\varepsilon$, $\delta$, $n$ be as in the statement of 
the theorem, and let $X$ and $S$ denote the stream and sample, respectively. We 
start with the Bernoulli sampling case, and assume that $p \geq 10\cdot 
\frac{\ln(4/ \delta) + \ln|\mathcal{R}|}{\varepsilon^2 n} = 10 \cdot 
\frac{\ln(4 |\mathcal{R}|/ \delta)}{\varepsilon^2 n}$.
For each $R \in \mathcal{R}$, we apply the first part of Lemma \ref{lem:main_step_upper_bound} with parameters $\varepsilon$ and $\delta / |\mathcal{R}|$, concluding that
$$
\Pr(|d_R(X) - d_R(S)| \geq \varepsilon) \leq \delta / |\mathcal{R}|.
$$
In the event that $|d_R(X) - d_R(S)| \leq \varepsilon$ for any $R$, by definition $S$ is an $\varepsilon$-approximation of $X$. Taking a union bound over all $R \mathcal{R}$, we conclude that the probability of this event \emph{not} to hold is bounded by $|\mathcal{R}| \cdot (\delta / |\mathcal{R}|) = \delta$, meaning that $\UniformSample$ with $p$ as above is $(\varepsilon, \delta)$-robust.

The proof for $\ReservoirSample$ is identical, except that we replace the condition on $p$ with the condition that $k \geq 2 \cdot \frac{ \ln(2/\delta) + \ln|\mathcal{R}|} {\varepsilon^2}$, and apply the second part of Lemma \ref{lem:main_step_upper_bound}.
\end{proof}

It is important to note that the typical proofs given for statements of this 
type in the static setting (i.e., when Adversary submits all elements in 
advance, and cannot act adaptively) do not apply for our adaptive setting. 
Indeed, the usual proof of the static analogue of the above lemma goes along 
the following lines: Adversary chooses which elements to submit in advance, and 
in particular, determines the number of elements from $A$ sent, call it $n_A$. 
Then, the number of sampled elements from $A$ is distributed according to the 
binomial distribution $\Bin(n_A, p)$ for Bernoulli sampling, and $\Bin(n_A, 
k/n)$ for reservoir sampling. One can then employ Chernoff bound to conclude 
the proof. This kind of analysis crucially relies on the adversary being static.

Here, we need to deal with an adaptive adversary. Recall that $\Adversary$ at any given point is modeled as a probabilistic process, that given the sequence $X_{i-1} =(x_1, \ldots, x_{i-1})$ of elements sent until now, and the current state $\sigma_{i-1}$ of $\Sampler$, probabilistically decides which element $x_i$ to submit next. Importantly, this makes for a well-defined probability space, and allows us to analyze $\Adversary$'s behavior with probabilistic tools, specifically with concentration inequalities.

Chernoff bound cannot be used here, as it requires the choices made by the adversary along the process to be independent of each other, which is clearly not the case. In contrast, martingale inequalities are suitable for this setting. We shall thus employ these, specifically Lemma \ref{lem:martingale-with-var}, to prove both parts of our main result in this section.

\subsection{The Bernoulli Sampling Case}
We start by proving the Bernoulli sampling case (first statement of Lemma 
\ref{lem:main_step_upper_bound}). Recall that here each element is sampled, 
independently, with probability $p$. At any given point $0 \leq i \leq n$ along 
the process, let $X_i = (x_1, \ldots, x_i)$ denote the sequence of elements 
submitted by the adversary until round $i$, and let $S_i \subseteq X_i$ denote 
the subsequence of sampled elements from $X_i$.  Note that $X_n = X$ and $S_n = 
S$, and hence, to prove the lemma, we need to show that $|d_R(X_n) - d_R(S_n)| 
\leq \varepsilon$.

As a first attempt, it might make sense to try applying a martingale concentration inequality on the sequence of random variables $(Y_0, Y_1, \ldots, Y_n)$, where we define $Y_i = d_R(X_i) - d_R(S_i)$. Indeed, our end-goal is to bound the probability that $Y_n$ significantly deviates from zero. However, a straightforward calculation shows that this is not a martingale, since the condition that $E[Y_i | Y_0, \ldots, Y_{i-1}] = 0$ does not hold in general. To overcome this, we show that a slightly different formulation of the random variables at hand does yield a martingale.
Given the above $R\subseteq U$, for any $0 \leq i \leq n$ we define the random variables
\begin{align}
A^R_i = \frac{i}{n} \cdot d_R(X_i) = \frac{|R\cap X_i|}{n}\qquad;\qquad B^R_i = \frac{|R\cap S_i|}{np} \qquad; \qquad Z^R_i = B^R_i - A^R_i,
\end{align} 
where, as before, the intersection between a set $R$ and a sequence $X_i$ is the subsequence of $X_i$ consisting of all elements that also belong to $R$.

Importantly, as is described in the next claim, the sequence of random variables $Z^R = (Z^R_0, \ldots, Z^R_n)$ defined above forms a martingale. The claim also demonstrates several useful properties of these random variables, to be used later in combination with Lemma \ref{lem:martingale-with-var}.

\begin{claim}
	\label{claim:uniform-martingale}
The sequence $(Z^R_0, Z^R_1, \ldots, Z^R_n)$ is a martingale. Furthermore, the variance of $Z^R_i$ conditioned on $Z^R_0, \ldots, Z^R_{i-1}$ is bounded by $1/n^2 p$, and it always holds that $|Z^R_i - Z^R_{i-1}| \leq 1/np$.
\end{claim}
We shall prove Claim \ref{claim:uniform-martingale} later on; first we use it to complete the proof of the main result. 
\begin{proof}[Proof of Lemma \ref{lem:main_step_upper_bound}, Bernoulli 
sampling case]
It suffices to prove the following two inequalities for any $p$ satisfying the 
conditions of the lemma for the Bernoulli sampling case:
\begin{align}
\label{eq:union-bound-uniform}
\Pr(|A^R_n - B^R_n| \geq \varepsilon / 2) \leq \delta / 2 \qquad ; \qquad \Pr(|B^R_n - d_R(S_n)| \geq \varepsilon / 2) \leq \delta/2 . 
\end{align}
Indeed, taking a union bound over these two inequalities, applying the triangle inequality, and observing that $A^R_n = d_R(X_n)$, we conclude that $\Pr(|d_R(X_n) - d_R(S_n)| \geq \varepsilon) \leq \delta$, as desired.

The first inequality follows from Claim \ref{claim:uniform-martingale} and Lemma \ref{lem:martingale-with-var}. Indeed, in view of Claim \ref{claim:uniform-martingale}, we can apply Lemma \ref{lem:martingale-with-var} on $(Z^R_0, \ldots, Z^R_n)$ with parameters $\lambda = \varepsilon/2$, $\sigma_i^2 = 1/n^2 p$, and $M = 1/np$. As $Z^R_0  = 0$, we have $|A^R_n - B^R_n| = |Z^R_n - Z^R_0|$, and so
$$
\Pr(|A^R_n - B^R_n| \geq \varepsilon / 2) \leq 2\exp\left(-\frac{(\varepsilon/2)^2}{2n \cdot \frac{1}{n^2 p} + \frac{\varepsilon}{6np}}\right) < 2\exp\left(-\frac{\varepsilon^2 n p}{9}\right). 
$$
The right hand side is bounded by $\delta / 2$ when $np \geq 
\frac{9}{\varepsilon^2} \ln(\delta / 4)$, settling the first inequality of 
\eqref{eq:union-bound-uniform}.

We next prove the second inequality of \eqref{eq:union-bound-uniform}. Observe that $B^R_n = d_R(S_n) \cdot \frac{|S_n|}{np}$.
Since each element is added to the sample with probability $p$, independently of other elements, the size of $S_n$ is distributed according to the binomial distribution $\Bin(n, p)$, regardless of the adversary's strategy. Applying Chernoff inequality with $\delta = \varepsilon / 2$, we get that 
$$
\Pr(\big| |S_n| - np\big| \geq \varepsilon np / 2) \leq 2 \exp\left(- \frac{\left(\varepsilon/2\right)^2 np}{2+\varepsilon/3}  \right) <  2 \exp\left( - \frac{\varepsilon^2 np}{10}  \right).
$$
This probability is bounded by $\delta / 2$ provided that $np \geq \frac{10 \ln (4/\delta)} {\varepsilon^{2}}$. Conditioning on this event not occurring, we have that 
$$  
\big|d_R(S_n) - B^R_n\big| = \bigg|1 - \frac{|S_n|}{np}\bigg|  \cdot d_R(S_n) \leq \bigg| 1 - \frac{|S_n|}{np}\bigg| \leq \frac{\varepsilon}{2}\ , 
$$
where the first inequality follows from the fact that densities (in this case, $d_R(S_n)$) are always bounded from above by one, and  the second inequality follows from our conditioning. This completes the proof of the second inequality in \eqref{eq:union-bound-uniform}.
\end{proof}

The proof of Claim \ref{claim:uniform-martingale} is given next.

\begin{proof}[Proof of Claim \ref{claim:uniform-martingale}]
We first show that $(Z^R_0, Z^R_1, \ldots, Z^R_n)$ is a martingale. Fix $1 \leq i  \leq n$, and suppose that the first $i-1$ rounds of $\AdaptiveGame_{\varepsilon}$ have just ended (so the values of $Z^R_0, \ldots, Z^R_{i-1}$ are already fixed), and that $\Adversary$ now picks an element $x_i$ to submit in round $i$ of the game. 

If $x_i \notin R$ then $A^R_i = A^R_{i-1}$ and $B^R_i = B^R_{i-1}$ and 
so 
$Z^R_i = Z^R_{i-1}$, which trivially means that $\E{Z^R_i\  | \ Z^R_0, 
\ldots, 
Z^R_{i-1}\ ;\ x_i \notin R} = Z^R_{i-1}$ as desired.

When $x_i \in R$, we have 
\begin{align*}
A^R_i = A^R_{i-1} + \frac{1}{n} \qquad ;& \qquad B^R_i = \left\{
\begin{array}{cl} B^R_{i-1} & \text{if } x_i \text{ is not sampled.}\\ B^{R}_{i-1} + \frac{1}{np} & \text{if } x_i \text{ is sampled.}
\end{array}
\right.\\
\Rightarrow & \qquad
Z^R_i = \left\{
\begin{array}{cl} Z^R_{i-1} - 1/n & \text{if } x_i \text{ is not sampled.}\\ Z^{R}_{i-1} + 1/np - 1/n & \text{if } x_i \text{ is sampled.}
\end{array}
\right.
\end{align*} 
Recall that $\Sampler$ uses Bernoulli sampling with probability $p$, that is, 
$x_i$ is sampled with probability $p$ (regardless of the outcome of the 
previous rounds). Therefore, we have that
$$\E{Z^R_i\  |\  Z^R_0, \ldots, Z^R_{i-1}\ ;\  x_i \in R} = Z^R_{i-1} + 
p \cdot 
(\frac{1}{np} - \frac{1}{n}) + (1-p) \cdot (-\frac{1}{n}) = Z^R_{i-1}.$$ 
The analysis of both cases $x_i \notin R$ and $x_i \in R$ implies that $E[Z^R_i | Z^R_0, \ldots, Z^R_{i-1}] = Z^R_{i-1}$, as desired.

We now turn to prove the other two statements of Claim \ref{claim:uniform-martingale}. 
The maximum of the expression $|Z^R_i - Z^R_{i-1}|$ is $\max\{\frac{1}{n}, \frac{1}{np} - \frac{1}{n} \} \leq \frac{1}{np}$, obtained when $x_i \in R$. The variance of $Z^R_i$ given $Z^R_0, \ldots, Z^R_{i-1}$ is zero given the additional assumption that $x_i \notin R$; assuming that $x_i \in R$, the variance satisfies
$$
\Var(Z^R_i \ |\  Z^R_0, \ldots, Z^R_{i-1}\ ;\  x_i \in R) = (1-p) \cdot 
\left(\frac{1}{n}\right)^2 + p \cdot \left( \frac{1}{np} - \frac{1}{n} 
\right)^2 = \frac{1}{n^2} \left(\frac{1}{p} - 1\right) \leq \frac{1}{n^2 p}.
$$
Combining both cases, we conclude that $\Var(Z^R_i \ |\  Z^R_0, \ldots, Z^R_{i-1}) \leq \frac{1}{n^2 p}$, completing the proof.
\end{proof}

\subsection{The Reservoir Sampling Case}
We continue to the proof of the second statement of Lemma 
\ref{lem:main_step_upper_bound}, which considers reservoir sampling. In high 
level, the proof goes along the same lines, except that we work with a 
different martingale. Specifically, for $k < i \leq n$ we define
\begin{align*}
A^R_i &= i \cdot d_R(X_i) = |R \cap X_i|, \\ B^R_i &= i \cdot d_R(S_i)
= \frac{i}{k} \cdot |R \cap S_i|, \\
 Z^R_i &= B^R_i - A^R_i,
\end{align*} 

whereas for $i \leq k$ we simply define $A^R_i = B^R_i = |R \cap X_i|$. (This 
is a natural extension of the definition for $i > k$; specifically, in view of 
the definition of $B^R_i$, note that as long as no more than $k$ elements 
appear in the stream, the reservoir simply keeps all of the stream's elements.)

The following claim is the analogue of Claim \ref{claim:uniform-martingale} for 
the setting of reservoir sampling.

\begin{claim}
	\label{claim:reservoir-martingale}
	The sequence $(Z^R_0, Z^R_1, \ldots, Z^R_n)$ is a martingale. Furthermore, 
	the variance of $Z^R_i$ conditioned on $Z^R_0, \ldots, Z^R_{i-1}$ is 
	bounded by $i/k$, and it always holds that $|Z^R_i - Z^R_{i-1}| \leq i/k$.
\end{claim}
\begin{proof}
We follow the same kind of analysis as in Claim \ref{claim:uniform-martingale}. 
Fix $i > k$ (for $i \leq k$ the claim holds trivially), and suppose that the 
first $i-1$ rounds have ended, so $Z^R_0, \ldots, Z^R_{i-1}$ are already fixed. 
Denote the next element that the adversary submits by $x_i$. 
First, it is easy to verify that
$$
A^R_i = \left\{ \begin{array}{cl} A^R_{i-1} & x_i \notin R \\ A^{R}_{i-1} + 1 & 
x_i \in R
\end{array}
\right. 
$$
The calculation of $B^R_i$ requires a more subtle case analysis. Given $B^R_0, 
\ldots, B^R_{i-1}$ and $x_i$, the value of $B^R_i$ is determined by three 
factors: (i) is $x_i \in R$ or not? (ii) is $x_i$ sampled or not? and (iii) 
conditioning on $x_i$ being sampled, does it replace an element from $R$ in the 
sample, or an element not in $R$? We separate the analysis into several cases; 
in cases where $x_i$ is sampled, we denote the element removed from the sample 
to make room for $x_i$ by $r_i$.

\paragraph{Case 1: $x_i \notin R$.} In the cases where $x_i$ is either not 
	sampled, or sampled but with $r_i \notin R$, elements from $R$ are neither added nor removed from the sample. That is, $R \cap S_i = R \cap S_{i-1}$. Hence,
	$$B^R_i = \frac{i}{k}\cdot |R \cap S_i| = \frac{i-1}{k} \cdot |R \cap S_{i-1}| + \frac{1}{k} \cdot |R \cap S_{i-1}| = B^R_{i-1} + d_R(S_{i-1}),$$
	where the first equality is by definition, and the third equality follows again by definition and since $|S_{i-1}| = k$ for $i > k$.
	
	It remains to consider the event where $x_i$ is sampled and $r_i 
	\in R$. The probability that $x_i$ is sampled equals $k / i$, and conditioning on this occurring, the probability that $r_i$ belongs to $R$ is $d_R(S_{i-1})$, so the above event holds with probability $(k/i) \cdot d_R(S_{i-1})$. In this case, one element from $R$ is removed from the sample, that is, $|R \cap S_i| = |R \cap S_{i-1}| - 1$, and therefore 
	$$B^R_i = \frac{i}{k} \cdot |R \cap S_i| = \frac{i}{k} \cdot |R \cap S_{i-1}| - \frac{i}{k} = B^R_{i-1} + d_R(S_{i-1}) - \frac{i}{k}.$$ 
	Thus, conditioned on $x_i \notin R$, the expectation of $B^R_i$ is  
	\begin{align*}
	 \left(1 - \frac{k}{i} \cdot d_R(S_{i-1})\right) \cdot \left(B^R_{i-1} + 
	d_R(S_{i-1})\right) + 
	 \frac{k}{i} \cdot d_R(S_{i-1}) \cdot \left(B^R_{i-1} 
	+ d_R(S_{i-1}) - \frac{i}{k}\right) 
	 = B^R_{i-1}.
	\end{align*} 
	Since $A^R_i = A^R_{i-1}$ when $x_i \notin R$, we deduce that
	$$
	E[Z^R_i | Z^R_0, \ldots, Z^R_{i-1}\ ;\ x_i \notin R] = Z^R_{i-1}.
	$$ 
	
	\paragraph{Case 2: $x_i \in R$.} Similarly, whenever $S_i = S_{i-1}$ we 
	have 
	that $B^R_i = B^R_{i-1} + d_R(S_{i-1})$. The only case where this does not 
	hold is when $x_i$ is sampled and $r_i \notin R$, which has probability 
	$(k/i) \cdot (1-d_R(S_{i-1}))$. In this case, $|R \cap S_i| = |R \cap 
	S_{i-1}| + 1$, implying that $$B^R_i = \frac{i}{k} \cdot |R \cap S_i| = \frac{i}{k} \cdot |R \cap S_{i-1}| + \frac{i}{k} = 
	B^R_{i-1} + d_R(S_{i-1}) + \frac{i}{k}.$$ 
	Combining these two we get, conditioned on 
	$x_i \in R$, that the expectation of $B^R_i$ is
		\begin{align*}
 	B^R_{i-1} + d_R(S_{i-1}) + \left( \frac{k}{i} \cdot \left(1 - 
 	d_R(S_{i-1})\right)\right) \cdot \frac{i}{k} = B^R_{i-1}+1.
	\end{align*} 
		Finally, since $A^R_i = A^R_{i-1} + 1$ when $x_i \in R$, we have that
		$$
		E[Z^R_i | Z^R_0, \ldots, Z^R_{i-1}\ ;\ x_i \in R] = Z^R_{i-1}.
		$$
The analysis of these two cases implies that $(Z^R_0, \ldots, Z^R_n)$ is indeed 
a martingale.

It remains to obtain the bounds on the difference $|Z^R_i - Z^R_{i-1}|$ and the 
variance of $Z^R_i$ given $Z^R_0, \ldots, Z^R_{i-1}$. This follows rather 
easily as a byproduct of the above analysis (and the fact that the density 
$d_R$ is always bounded between zero and one). When $x_i \notin R$, we know 
from the analysis that $A^R_i = A^R_{i-1}$ and $B^R_{i-1} - i/k \leq B^R_i \leq 
B^R_{i-1} + 1$, whereas if $x_i \in R$, we have $A^R_i = A^R_{i-1} + 1$ and 
$B^R_{i-1} \leq B^R_i \leq B^R_{i-1} + 1 + i/k$. In both cases, we conclude 
that $|Z^R_i - Z^R_{i-1}| \leq i/k$. 

We next bound the variance of $Z^R_i$ conditioned on the values of $Z^R_0, \ldots, 
Z^R_{i-1}$ (the analysis also implicitly conditions on the value 
$d_R(S_{i-1})$; the bound we shall eventually derive holds regardless of this 
value). We start with the case that $x_i \notin R$, and revisit Case 1 above: 
with probability $(k / i) \cdot d_R(S_{i-1})$, the value of $Z^R_i$ is smaller 
than its expectation by $i/k - d_R(S_{i-1})$; and otherwise (with probability 
$1 - (k / i) \cdot d_R(S_{i-1})$), the value of $Z^R_i$ is larger than its 
expectation by $d_R(S_{i-1})$. 
Thus, we have that
\begin{align*}
\Var&(Z^R_i\ |\ Z^R_0, \ldots, Z^R_{i-1},\  x_i \notin R,\  d_R(S_{i-1})) \\
& = 
\frac{k}{i} \cdot d_R(S_{i-1}) \cdot \left(\frac{i}{k} - 
d_R(S_{i-1})\right)^2 + \left( 1 - \frac{k}{i} \cdot d_R(S_{i-1}) \right) 
\cdot \left(d_R(S_{i-1})\right)^2 \\
& =\frac{i}{k} \cdot d_R(S_{i-1}) - \left(d_R(S_{i-1})\right)^2 \leq 
\frac{i}{k}.
\end{align*}
We next address the case where $x_i \in R$, which correspond to Case 2 above. 
Here, with probability $(k/i) \cdot (1-d_R(S_{i-1}))$, the value of $Z^R_i$ is larger than its conditional expectation by $i/k + d_R(S_{i-1})- 1$; otherwise, $Z^R_i$ is smaller than the expectation by $1 - d_R(S_{i-1})$. Thus,
\begin{align*}
 \Var&(Z^R_i\ |\ Z^R_0, \ldots, Z^R_{i-1},\  x_i \in R,\  d_R(S_{i-1})) \\
& = \frac{k}{i} \cdot \left(1 - d_R(S_{i-1})\right) \cdot \left(\frac{i}{k} + 
d_R(S_{i-1}) - 1\right)^2 + 
 \left( 1 - \frac{k}{i} \cdot \left(1 - d_R(S_{i-1})\right) \right) \cdot 
\left(1 - d_R(S_{i-1})\right)^2 \\
&= \frac{i}{k} \cdot \left(1 - d_R(S_{i-1})\right) - \left(1 - d_R(S_{i-1})\right)^2 \leq \frac{i}{k}.
\end{align*}
 As the conditional variance is always bounded by $i/k$, the bound remains intact if we remove the conditioning on the value of $d_R(S_{i-1})$ and the predicate assessing whether $x_i \in R$ or not. In other words, $\Var(Z^R_i | Z^R_0, \ldots, Z^R_{i-1}) \leq i/k$, completing the proof.
\end{proof}

The proof of the second part of Lemma \ref{lem:main_step_upper_bound} now 
follows from the last claim.

\begin{proof}[Proof of Lemma \ref{lem:main_step_upper_bound}, reservoir 
sampling case]
Observe that
\begin{align*}
\Pr(|d_R(X) - d_R(S)| \geq \varepsilon) &= \Pr(|B^R_n - A^R_n| \geq \varepsilon  
n) \\ &= \Pr(|Z^R_n - Z^R_0| \geq \varepsilon n).
\end{align*}
In view of Claim \ref{claim:reservoir-martingale}, we apply Lemma 
\ref{lem:martingale-with-var} on the martingale $Z^R = (Z^R_0, \ldots, Z^R_n)$ with 
$\lambda = \varepsilon n$, $\sigma^2_i = i / k$ for any $i \geq k$ (for $i \leq 
k$, we can set $\sigma^2_i = 0$), and $M = n/k$. We get that
\begin{align*}
\label{eq:prob_bound_reservoir}
\Pr(|Z^R_n - Z^R_0| \geq \lambda) &\leq 2\exp \left( - \frac{\lambda^2}{2 
\sum_{i=1}^n \sigma_i^2 + M\lambda/3}  \right) \\
& = 2\exp \left( - \frac{\varepsilon^2 n^2}{2 \sum_{i=1}^n (i/k) + (n/k) \cdot 
\varepsilon n /3}  \right) \\
&= 2\exp \left( - \frac{\varepsilon^2 k n^2}{n(n+1) + \varepsilon n^2 /3}  
\right) \hspace{0.13cm} \\
& \leq 2 \exp\left(-\frac{\varepsilon^2 k n^2}{2n^2}\right) = 2 \exp\left(-\frac{\varepsilon^2 k}{2}\right),
\end{align*}
where the second inequality holds for $n \geq 2$.
Therefore, it suffices to require $k \geq \frac{2}{\varepsilon^2} \ln\left(\frac{2}{\delta}\right)$ 
to get the bound $\Pr(|d_R(X) - d_R(S)| \geq \varepsilon) \leq \delta$.
\end{proof}

\section{An Adaptive Attack on Sampling}
\label{sec:attack}
In this section, we present our lower bounds. Specifically, we show that the 
sample size cannot depend solely on the VC-dimension, but rather that the 
dependency on the cardinality is necessary. This is done by describing a set 
system $(U,\mathcal{R})$ with large $|U|$ and VC-dimension of one, together with a 
strategy for the adversary that will make the sampled set 
unrepresentative with respect to $(U,\mathcal{R})$.
That is, the sampled set will not be an $\varepsilon$-approximation of 
$(U,\mathcal{R})$ with high 
probability. This is in contrast to the static setting where the same sample size 
suffices to an $\varepsilon$-approximation with high probability.
Moreover, in 
the case of the $\UniformSample$ algorithm, the sampled set under attack is extremely 
unrepresentative, consisting precisely of the $k$ smallest elements in the stream (where $k$ is the total sample size at the end of the stream).
\begin{proof}[Proof of \Cref{thm:attack}]
Set the universe to be the well-ordered set $U=\{1,2,\ldots,N\}$ for an 
arbitrary $n^{6\ln n} \leq N \leq 2^{n/2}$ and 
let  $\mathcal{R}=\{[1,b]:b \in U\}$. Clearly, 
$(U,\mathcal{R})$ has VC-dimension 1.
$\Adversary$'s strategy (for both sampling algorithms $\UniformSample$ and 
$\ReservoirSample$) is described in 
\Cref{fig:attack2}.
\protocol
{The adversarial strategy}
{The description of $Adversary$'s strategy for making the sample 
unrepresentative.}
{fig:attack2}
{
\begin{enumerate}
	\item Set $a_1=1$ and $b_1=N$.
	\item Let $p'=\max\{p,\ln n/n\}$.
	\item For $i=1\ldots n$ do:
	\begin{enumerate}
		\item Set $x_i=\lfloor a_i + (1-p')(b_i-a_i) \rfloor$.
		\item If $x_i$ is sampled then set $a_{i+1}=x_i$ and $b_{i+1}=b_i$.
		\item Otherwise set $a_{i+1}=a_i$ and $b_{i+1}=x_i$.
	\end{enumerate}
	\item The final stream is $X=x_1,\ldots,x_n$.
\end{enumerate}
}

Let $S$ denote the subsequence of elements sampled by the algorithm 
$\UniformSample$ along the stream. The expected size of $S$ is $np \leq np'$, 
and it follows from the well-known Markov inequality (see e.g. \cite{AlonS92}, 
Appendix A) that $\Pr(|S| \geq 2np') < 1/2$ (in fact the probability is much 
smaller, by Chernoff inequality, but we will not need the stronger bound). From 
here on, we condition on the complementary event: we assume that $|S| < 2np'$.
The next claim asserts that for $S$ of this size, Adversary's strategy does not 
fail, in the sense that it never runs out of elements (i.e., $a_i < b_i$ for 
all $i \in [n]$).
\begin{claim}\label{clm:big-range}
If $|S| < 2np'$ then $b_i - a_i \geq n$ for any $i \in [n]$.
\end{claim}
\begin{proof}
For any $i \in [n]$, set $\ell_i = b_i-a_i$. We prove by induction that $\ell_i \geq n$.
If $x_i$ is sampled, then we have that 
$\ell_{i+1} \geq p'\ell_i$ and otherwise we have that $\ell_{i+1}=(1-p')\ell_i -2  \geq   (1-2p') \ell_i$, where the inequality follows from the induction assumption. Since $|S| < 2np'$, we get that 
\begin{align*}
\ell_i &\ge\nonumber p'^{|S|}(1-2p')^{n-|S|} \cdot N \nonumber \\ 
& \ge p'^{|S|}(1-2p')^n \cdot N \nonumber \\
&  = e^{-(|S|\ln \frac{1}{p'} + n\ln(\frac{1}{1-2p'}))} \cdot N \nonumber \\ 
& > e^{-(|S|\ln \frac{1}{p'} + 3np')} \cdot N  \\
& > e^{-(2np'\ln \frac{1}{p'} + 3np')} \cdot N \\ 
& \ge e^{\ln n-\ln N} \cdot N  = n~,
\end{align*}
where the third inequality holds since $3p \ge 
\ln(\frac{1}{1-2p})$ for small 
enough $p > 0$, and the last inequality follows since $p' \leq \frac{\ln N}{6n\ln 
n}$ and $p' \ge \ln n/n$, which means that
$$
\ln N \ge 6np'\ln n \ge 2np'\ln (1/p') + 3np' + \ln n.
$$
This proves the induction step, and completes the proof of the claim.
\end{proof}
The last claim means that if $|S| < 2np'$, then the attack in 
\Cref{fig:attack2} successfully generates a stream of $n$ elements.
We now show that the sampled 
set is not an $\varepsilon$-approximation. 
We begin by analyzing  the $\UniformSample$ algorithm.
\begin{claim}
Consider $\Adversary$'s attack on $\UniformSample$ described in \Cref{fig:attack2}. At round $i$ of the game, 
\begin{itemize}
	\item All elements that were previously submitted by $\Adversary$ and \emph{sampled} are no bigger than $a_i$.
	\item All elements that were previously submitted but \emph{not sampled} are no smaller than $b_i$.
	\item The element submitted during round $i$ is between $a_i$ and $b_i$.
\end{itemize}
\end{claim}
\begin{proof}
By induction, where the base case $i=1$ is trivial. Suppose that the claim 
holds for the first $i-1$ rounds; we now prove it for round $i$. By definition 
of the attack, and from Claim~\ref{clm:big-range} it holds that $a_{i-1} \leq 
a_i < 
b_i 
\leq b_{i-1}$ and so any of the elements 
$x_j$ for $j < i-1$ satisfies the desired condition, by the induction 
assumption. It remains to address the case where $j = i-1$. If $x_{i-1}$ was 
sampled, then the attack sets $a_i = x_{i-1}$, that is, $x_{i-1}$ is a sampled 
element and satisfies $x_{i-1} \leq a_i$. Otherwise, the attack sets $b_i = 
x_{i-1}$ and so $x_{i-1}$ is a non-sampled element and satisfies $x_{i-1} \ge 
b_i$. Finally, $a_i < x_i < b_i$ always holds. Thus, the three desired conditions are retained.
\end{proof}

%
%
%
As the last claim depicts, all \emph{sampled} elements are smaller than all 
\emph{non-sampled} ones at any point along the stream.
This, of course, suffices for the sampled set to not be an 
$\varepsilon$-approximation of $(U,\mathcal{R})$. Denote the sampled set by $S$, 
and let $s$ be the maximal element in $S$ (if $S$ is empty, we are done). Consider now the range 
$[1,s] \in 
\mathcal{R}$: its density in the sampled set is $1$, namely, 
$d_{[1,s]}(S)=1$, while its density in the stream is $d_{[1,s]}(X)=|S|/n$. 
To summarize,
$$
|d_{[1,s]}(S) - d_{[1,s]}(X)| \ge
1-|S|/n \ge
1-2p' > 
1/2 \ge \varepsilon~.
$$
Altogether, the attack does not fail provided that $|S| < 2np'$, which holds with probability at least $1/2$. Thus, $\UniformSample$ with parameter $p$ as in the  theorem's statement is not $(\varepsilon, \delta)$-robust.

The analysis of the $\ReservoirSample$ algorithm is very similar. Recall that 
$k$ denotes the sample size, and let $k'$ be the total number of elements that were 
sampled 
during the reservoir sampling process. That is, $k'$ counts sampled elements 
that were evicted at a future iteration. We bound $k'$ as follows. $\E{k'}=k 
+ \sum_{i =1}^{n}k/n \le 2k\ln n$. Again, Markov inequality shows that with 
probability at least $1/2$, we will have $k' \le 4k\ln n$. Using the previous analysis, we 
know that all $k'$ elements are the smallest elements in the stream. The sample 
set $S$ consists of some $k$ elements among these $k'$ elements (in other 
words, the sample set is not necessarily the set of $k$ smallest element, but 
it is still a subset of the $k'$ smallest elements). Thus, taking the 
interval $[1,s]$ where $s$ is the maximal element among the $k'$ elements, we 
have that the density of $[1,s]$ in the sample is $ d_{[1,s]}(S) = 
\frac{k}{k}=1$.
On the 
other hand, the density of $[1,s]$ is the stream is
$$
d_{[1,s]}(X) = \frac{k'}{n} \le \frac{4k\ln n}{n} \leq \frac{\ln N}{n} \leq 1/2. 
$$
Together, we entail that
$$
|d_{[1,s]}(S)-d_{[1,s]}(X)| > 1-1/2 \geq \varepsilon~,
$$ meaning that $\ReservoirSample$ with $k$ as in the statement of the theorem is not $(\varepsilon, \delta)$-robust.
\end{proof}


\section{Continuous Robustness}
\label{sec:continuous}

In this section, we prove that the $\ReservoirSample$ algorithm is 
$(\varepsilon, \delta)$-{\em continuous} robust against static and adaptive 
adversaries. Recall that a sampling algorithm is $(\varepsilon, 
\delta)$-continuously robust if the following holds with probability at least 
$1- 
\delta$: at \emph{any} point throughout the stream, the current sample held by 
$\Sampler$ is an $\varepsilon$-approximation of the current stream (i.e., of the set of all elements submitted by $\Adversary$ until now).

With this definition in hand, $\UniformSample$ cannot possibly be continuously 
robust in general (even in the static setting)\footnote{To see this, consider 
any set system $(U, \mathcal{R})$ where $\mathcal{R}$ contains a singleton 
$\{u\}$ for some $u \in U$, which is the first element of the stream. With 
probability $1-p$ this element is not sampled and the density of $\{u\}$ in 
the sample at the current point is $0$, while its density in the stream is $1$. 
This violates the $\varepsilon$-approximation requirement (unless $p \geq 
1-\delta$).}.
We thus restrict our discussion to $\ReservoirSample$ from here on, and turn to 
the proof of Theorem \ref{thm:continuous}. The proof examines 
$O(\varepsilon^{-1} \ln n)$ carefully picked points along the stream, applying 
Theorem \ref{thm:UB} on each of the points. It then shows that if the sample  
is a good approximation of the stream at all of these points, then continuous 
robustness is guaranteed with high probability.
 
\begin{proof}[Proof of Theorem \ref{thm:continuous}]
We provide the proof for the setting of an \emph{adaptive} adversary. The proof for the static setting is essentially identical, with the only difference being that, instead of making black-box applications of Theorem \ref{thm:UB}, we apply the static analogue of it; Recall that the bound in the static analogue is of the form $\Theta\left(\frac{d + \ln 1/\delta}{\varepsilon^2}\right)$, compared to the $\Theta\left(\frac{\ln|\mathcal{R}| + \ln 1/\delta}{\varepsilon^2}\right)$ bound appearing in the statement of Theorem \ref{thm:UB}.
	
Let $(U, \mathcal{R})$, $n$, $\varepsilon$, $\delta$ be as in the statement of 
the theorem. 
As a warmup, let us analyze a simple yet non-optimal proof based on a na\"ive 
union bound.
Denote the stream and sample after $i$ rounds by $X_i$ and $S_i$, respectively.
Consider for a moment the first $i$ rounds of the game as a ``standalone'' game where the stream length is $i$.
Applying the second part of Theorem \ref{thm:UB} with parameters $(U, \mathcal{R}), \varepsilon, \delta', i$, where $\delta' = \delta/n$, we get that if the memory size $k$ of $\ReservoirSample$ satisfies 
\begin{equation}
\label{eq:naive_bound_k}
k \geq 2 \cdot \frac{\ln |\mathcal{R}| + \ln(2n / \delta)}{\varepsilon^2} = 
2 \cdot \frac{\ln |\mathcal{R}| + \ln(2 / \delta) + \ln n}{\varepsilon^2},
\end{equation}
then regardless of $\Adversary$'s strategy, 
$$
\Pr(\text{$S_i$ is not an $\varepsilon$-approximation of $X_i$}) \leq \delta/n. 
$$
Taking a union bound, the probability that $S_i$ is an $\varepsilon$-approximation of $X_i$ for all $i \in [n]$ is at least $1 - n \cdot (\delta / n) = 1 - \delta$. Thus, it follows that $\ReservoirSample$ whose parameter $k$ satisfies the condition of \eqref{eq:naive_bound_k} is $(\varepsilon, \delta)$-continuously robust.

We now continue to the proof of the improved bound, appearing in the statement of the theorem. 
The proof is also based, at its core, on a union bound argument, albeit a more efficient one. The key idea is to take a sparse set of ``checkpoints'' $i_1, \ldots, i_t$ along the stream, where $i_{j+1} = (1+\Theta(\varepsilon)) i_j$, apply Theorem \ref{thm:UB} at any of the times $i_1, \ldots, i_t$ to make sure the sample is an $(\varepsilon/2)$-approximation of the stream in any of these times. Finally, we show that with high probability, for any $j \in [t-1]$, the approximation is preserved (the approximation factor might become slightly worse, but no worse than $\varepsilon$) in the ``gaps'' between any couple of such neighboring points.

For this, we first need the following simple claims.
\begin{claim}
	\label{claim:approx_of_approx}
Let $T, T'$ be two sequences of length $k$ over $U$, which differ in up to $v$ values. Then $|d_R(T) - d_R(T')| \leq v/k$ for any $R \subseteq U$. 
In particular, if $T$ is an $\alpha$-approximation of some sequence $X \supseteq T, T'$, then $T'$ is an $(\alpha + v/k)$-approximation of $X$.
\end{claim}
\begin{proof} 
	For any subset $R \subseteq U$ we have $-v \leq |R \cap T| - |R \cap T'| \leq v$. Dividing by $k = |T| = |T'|$, and recalling that $d_R(T) = |R \cap T| / |T|$ and $d_R(T') = |R \cap T'| / |T'|$, we conclude that $-v/k \leq d_R(T) - d_R(T') \leq v/k$, that is, $|d_R(T) - d_R(T')| \leq v/k$. To prove the second part, note that $$|d_R(T') - d_R(X)| \leq |d_R(T') - d_R(T)| + |d_R(T) - d_R(X)| \leq v/k + \alpha$$ 
	for any $R \subseteq U$.
\end{proof}
\begin{claim}
	\label{claim:approx_larger_set}
Suppose that $T \subseteq X \subseteq X'$ are three sequences over $U$, where $T$ is an $\alpha$-approximation of $X$, and $|X'| \leq (1+\beta) |X|$. Then $T$ is an $(\alpha+\beta)$-approximation of $X'$.
\end{claim}
\begin{proof}
For any subset $R \subseteq U$, we have that $|R \cap X| \leq |R \cap X'| \leq |R \cap X| + \beta |X|$. We also know that $|d_R(T) - d_R(X)| \leq \alpha$, since $T$ is an $\alpha$-approximation of $X$.
On the one hand, it follows that 
\begin{align*}
 d_R(T) & \geq d_R(X) - \alpha = \frac{|R \cap X|}{|X|} - \alpha \geq \frac{|R 
\cap X'| - \beta|X|}{|X|} - \alpha \\
& \geq \frac{|R \cap X'|}{|X'|} - \beta - \alpha = d_R(X') - (\alpha + \beta).
\end{align*}
On the other hand,
\begin{align*}
d_R(T) & \leq d_R(X) + \alpha = \frac{|R \cap X|}{|X|} + \alpha \leq \frac{|R 
\cap X'|}{|X'| / (1+\beta)} + \alpha \\
& = (1+\beta)d_R(X') + \alpha \leq d_R(X') + (\alpha + \beta). 
\end{align*}
As these inequalities hold for any $R \subseteq U$, the claim follows.
\end{proof}
As a consequence of the above two claims, we get the following useful claim. (Recall that for any $i \in [n]$, the sample and stream after $i$ rounds are denoted by $S_i$ and $X_i$, respectively.)
\begin{claim}
	\label{claim:corollary_continuous}
Consider $\ReservoirSample$ with memory size $k$, and suppose that exactly $v$ elements were sampled in rounds $l+1, l+2, \ldots, m$ of the game, where $k \leq l < m \leq (1 + \beta)l$. If $S_l$ is an $\alpha$-approximation of $X_l$, then $S_m$ is an $(\alpha + \beta + v/k)$-approximation of $X_m$.
\end{claim}
\begin{proof}
By Claim \ref{claim:approx_larger_set}, $S_l$ is an $(\alpha + \beta)$-approximation of $X_m$. As $S_m$ differs from $S_l$ by at most $v$ elements, we conclude from Claim \ref{claim:approx_of_approx} that $S_m$ is an $(\alpha + \beta + v/k)$-approximation of $X_m$.
\end{proof}
The last claim equips us with an approach to ensure continuous robustness, which is more efficient compared to the simple union bound approach. Suppose that there exists a set of integers $k = i_1 < i_2 < \ldots < i_t = n$ satisfying the following for any $j \in [t-1]$.
\begin{enumerate}
\item $S_{i_j}$ is an $\alpha$-approximation of $X_{i_j}$, where $\alpha = \varepsilon/4$.
\item $i_{j+1} \leq (1+\beta) i_j$, where $\beta=\varepsilon/4$.
\item  The number of elements sampled in rounds $i_j+1, i_j+2, \ldots, i_{j+1}$ is bounded by $v = \varepsilon k / 2$.
\end{enumerate}
We claim that the above three conditions suffice to ensure that $S_i$ is an $\varepsilon$-approximation of $X_i$ for any $i \in [n]$. Indeed, for $i \leq k$, $S_i = X_i$ is trivially an $\varepsilon$-approximation. When $i > k$, consider the maximum $j < t$ for which $i_j \leq i$, and apply Claim \ref{claim:corollary_continuous} with $l = i_j$, $m = i$, and $\alpha, \beta, v$ as dictated above.
Since $\alpha+\beta+v/k=\varepsilon$, the claim implies that $S_i$ is an $\varepsilon$-approximation of $X_i$, as desired.

Specifically, given $k$ satisfying the assumption of Theorem \ref{thm:continuous}, we pick $i_1, i_2, \ldots, i_t$ recursively as follows: we start with $i_1 = k$; and given $i_j$ we set $i_{j+1} \leq n$ as the largest integer satisfying that $i_{j+1} \leq (1+\beta)i_j = (1 + \varepsilon/4) i_j$. 
It is not hard to verify that $i_j = k \cdot (1+\theta(\varepsilon))^{j-1}$ 
(this implicitly relies on the fact that $k \geq 4/\varepsilon$, ensured by the 
assumption of the theorem). Note that $t = O(\ln_{1+\varepsilon} n) = 
O(\varepsilon^{-1} \ln n)$. We next show that for this choice of $i_1, \ldots, 
i_t$, the above three conditions are satisfied simultaneously for all $j \in 
[t-1]$ with probability at least $1-\delta$. This shall conclude the proof.

For the first condition, apply Theorem \ref{thm:UB}  for any $j \in [t-1]$ with parameters $(U, \mathcal{R}), \varepsilon/42, \delta', i_j$ where $\delta' = \delta / 2t$, concluding that if the memory size $k$ satisfies
\begin{equation*}
k \geq 2 \cdot \frac{\ln |\mathcal{R}| + \ln(4t / \delta)}{(\varepsilon/4)^2} = 
\Theta \left(\frac{\ln |\mathcal{R}| + \ln(1 / \delta) + \ln(1 / \varepsilon) + \ln \ln n}{\varepsilon^2} \right)
\end{equation*}
then for any $j \in [t-1]$,
\begin{equation*}
\Pr(S_{i_j} \text{ is an $\varepsilon/2$-approximation of } X_{i_j}) \geq 1 - \delta / 2t.
\end{equation*}
Taking a union bound, with probability at least $1 - \delta/2$ the first condition holds for all $j \in [t-1]$.

The second condition, regarding the boundedness of $i_{j+1}$ as a function of $i_j$, holds trivially (and deterministically) for our choice of $i_1 \leq i_2 \leq \ldots \leq i_t$. 

Finally, it remains to address the third condition.
For any $j \in [t-1]$, let $A_j$ denote the total number of sampled elements in rounds $i_j+1, i_j+2, \ldots, i_{j+1}$ of the game. Note that each such $A_j$ is a random variable.  We wish to show that 
\begin{equation}
\label{eq:num_elements_sampled}
\Pr(A_j > \varepsilon k / 2) \leq \delta / 2t.
\end{equation}
Indeed, if \eqref{eq:num_elements_sampled} is true for any $j \in [t-1]$, then 
the probability that the third condition holds for any $j$ is at least $1 - 
\delta/2$, which (in combination with our analysis of the other two conditions) 
completes the proof. Thus, it remains to prove \eqref{eq:num_elements_sampled}.

Recall that the probability of an element to be sampled in round $i$ is exactly $k/i$, and that $i_{j+1} \leq (1+\varepsilon/4)i_j$. Hence, $A_j$ is a sum of up to $\lfloor \varepsilon i_j / 4 \rfloor$ independent random variables, each of which has probability less than $k / i_j$ to be sampled. In particular, the mean of $A_j$ is less than $(\varepsilon i_j / 4) \cdot (k / i_j) = \varepsilon k / 4$. From Chernoff bound (Theorem \ref{lem:chernoff}), we get the desired bound: 
$$
\Pr(A_j \geq \varepsilon k / 2) < \exp\left( - \frac{2^2 \cdot \varepsilon k / 4}{2 + 2\cdot 2 / 3}  \right) \leq \exp \left( - \frac{\varepsilon k}{4}\right) \leq \frac{\delta}{2t},
$$ 
where the last inequality holds for $k \geq c \cdot \varepsilon^{-1}(\ln 
1/\delta + \ln 1/\varepsilon + \ln \ln n)$, for a sufficiently large constant 
$c > 0$; note that $k$ in the theorem's statement indeed satisfies this 
inequality. 
\end{proof}

\section*{Acknowledgments}
We are grateful to Moni Naor for suggesting the study of streaming 
algorithms 
in the adversarial setting and for helpful and informative discussions about 
it. We 
additionally 
thank Noga Alon, Nati Linial, and Ohad Shamir for invaluable comments and 
suggestions for the paper.

\bibliographystyle{alpha}
\bibliography{adaptive-sampling}

\end{document}